\documentclass[twoside,leqno,twocolumn]{article}  
\usepackage{ltexpprt}
 
\usepackage{subfig}
\usepackage{graphicx}
\graphicspath{{./}{pics/}}
\usepackage{verbatim} 
\usepackage{amssymb}
\usepackage{color}
\usepackage{enumerate}
\usepackage{array}
\usepackage{footnote}
\usepackage[math]{cellspace}

\usepackage{hyperref}
\definecolor{darkgreen}{rgb}{0,0.4,0}
\definecolor{BrickRed}{rgb}{0.65,0.08,0}
\hypersetup{colorlinks=true,linkcolor=blue,citecolor=red,filecolor=BrickRed,urlcolor=darkgreen}

\newcommand{\exzk}[2]{[z^{#1}]_{#2}}
\newcommand{\exzkn}{\exzk{n}{\zeta_k}}
\newcommand{\coefb}{C_7}
\newcommand{\coefd}{C'_7}
\newcommand{\walksym}{\omega}
\newcommand{\walk}[1]{\walksym_{#1}}

\newcommand{\stepset}{\mathcal{S}}

\newcommand{\LandauO}{\mathcal{O}}

\newcommand{\Dc}{\mathcal{D}}

\newcommand{\C}{\mathbb{C}}
\newcommand{\N}{\mathbb{N}}

\newcommand{\Z}{{\mathbb Z}}

\begin{document}
\title{\Large Lattice paths of slope $2/5$}
\author{Cyril Banderier\thanks{Laboratoire d'Informatique de Paris Nord, UMR CNRS 7030, Universit\'e Paris Nord, 93430 Villetaneuse, France} \\
\and 
Michael Wallner\thanks{Institute of Discrete Mathematics and Geometry, TU Wien, Wiedner Hauptstr. 8-10/104, A-1040 Wien, Austria}
}
\date{\it{This article corresponds, up to minor typo corrections, to the extended abstract which appeared in  pp.~105--113  of the 
2015 Proceedings of the Twelfth Workshop on Analytic Algorithmics and Combinatorics (ANALCO),
a conference which held in San Diego in January 2015.  A long version of this work ``Lattice paths below a line of (ir)rational slope'' will appear soon.}
}

\maketitle

\begin{abstract} \small\baselineskip=10pt 
We analyze some enumerative and asymptotic properties of Dyck paths under a line of slope 2/5.
This answers to Knuth's problem \#4 from his ``Flajolet lecture'' during the conference ``Analysis of Algorithms'' (AofA'2014) in Paris in June 2014.
Our approach relies on the work of Banderier and Flajolet for asymptotics and enumeration of directed lattice paths. 

A key ingredient in the proof is the generalization of an old trick of Knuth himself 
(for enumerating permutations sortable by a stack),
promoted by Flajolet and others as the ``kernel method''.
All the corresponding generating functions are algebraic,
and they offer some new combinatorial identities, 
which can be also tackled in the {\em A=B} spirit of Wilf--Zeilberger--Petkov{\v s}ek.

We show how to obtain similar results for other slopes than 2/5, 
an interesting case being e.g.\ Dyck paths below the slope 2/3, 
which corresponds to the so called Duchon's club model.
\end{abstract}

\section{Introduction}
\label{sec:intro}

What Flajolet named the ``kernel method''
has been part of the folklore of combinatorialists for some time.
Earlier references usually deal with the case of 
a functional equation (with apparently more unknowns than equations!) of the form
\[
K(z,u)F(z,u)= A(z,u)+B(z,u) G(z)
\]
where $A, B$, and $K$ are given and where $F, G$ are the unknown functions we want to get.
$K(z,u)$ is a polynomial in $u$ which we call the ``kernel'' as
we ``test'' this functional equation on functions $u(z)$ canceling this kernel.
The simplest case is when there is only one branch, $u_1$, such that $K(z,u_1(z))=0$ and $u_1(0)=0$;
in that case, a single substitution gives a closed form solution for $G$:
namely, $G(z)=-A(z,u_1)/B(z,u_1)$.
One clear source of this is the detailed solution to Exercise 2.2.1--4 
in {\em the Art of Computer Programming}
(\cite[pp.~536--537]{Kn69} and also Ex.~2.2.1.11 therein),
which presents a ``new method for solving the ballot problem'', 
for which the kernel equation is quadratic (it then involves just one small branch $u_1$).

In combinatorics, there are many applications of variants of this way of solving functional equations:
e.g.\ it is known as the ``quadratic method'' in map enumeration, as initially done by Tutte and Brown, see also~\cite{bfss01},
or as the ``iterated kernel method'' in queuing theory (as done in~\cite{Fayolle} for non directed walks in the quarter plane), 
and it also plays a key r\^ole for other constraint lattice paths and their asymptotics~\cite{BaFl02}, for additive parameters~\cite{BaGi06}, 
for generating trees~\cite{hexa}, for avoiding-pattern permutations~\cite{Mansour}, for statistics in posets~\cite{Fusy}...

Let us give a definition of the lattice paths we consider:
 \begin{Definition} \label{def:LP}
 A {\it step set} $\stepset \subset \Z^2$, is a finite set of vectors $\{ (x_1,y_1), \ldots, (x_m,y_m)\}$. 
An $n$-step \emph{lattice path} or \emph{walk} is a sequence of vectors $v = (v_1,\ldots,v_n)$, such that $v_j$ is in $\stepset$. 
Geometrically, it may be interpreted as a sequence of points $\walksym =(\walk{0},\walk{1},\ldots,\walk{n})$ where $\walk{i} \in \Z^2, \walk{0} = (0,0)$ (or another starting point)
and $\walk{i}-\walk{i-1} = v_i$ for $i=1,\ldots,n$.
The elements of $\stepset$ are called \emph{steps} or \emph{jumps}. 
The \emph{length} $|\walksym|$ of a lattice path is its number $n$ of jumps. 
 \end{Definition}
The lattice paths can have different additional constraints shown in Table \ref{fig-4types}.
  \begin{table*}[t]
 \small
 \begin{center}\renewcommand{\tabcolsep}{3pt}
 \begin{tabular}{|c|c|c|}
 \hline
 & ending anywhere & ending at 0\\
 \hline
 \begin{tabular}{c} unconstrained \\ (on~$\Z$) \end{tabular}
 & \begin{tabular}{c} 

 {\includegraphics[width=5cm]{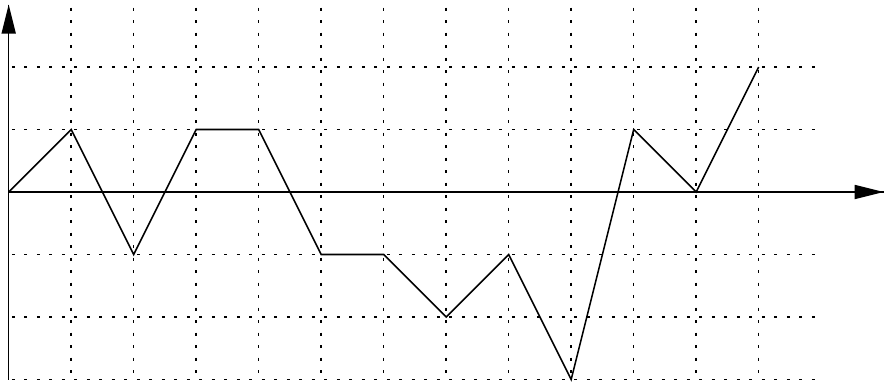}} 
\\ 
 walk/path ($\cal W$) 
 \end{tabular}
 & \begin{tabular}{c} 

 {\includegraphics[width=5cm]{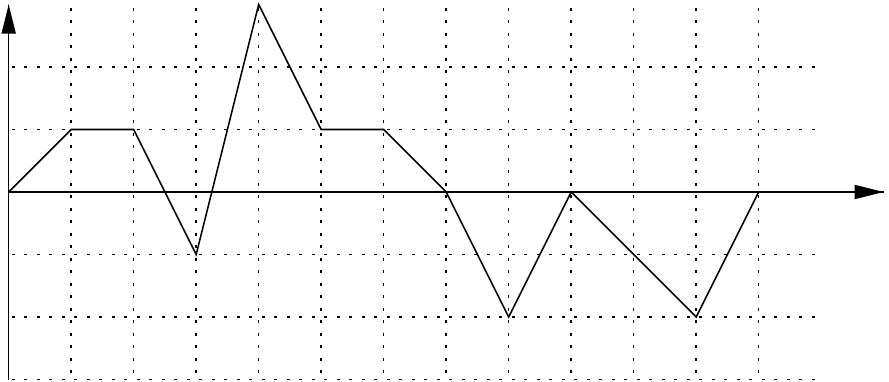}} 
\\
 bridge ($\cal B$)
 \end{tabular} \\
 \hline
 \begin{tabular}{c}constrained\\ (on $\N$) \end{tabular}
 & \begin{tabular}{c} 
 \includegraphics[width=5cm]{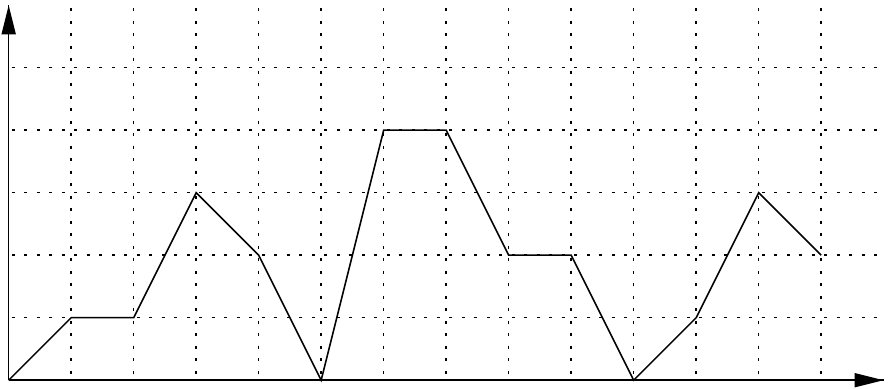} 
\\ 
 meander ($\cal M$)\\ 
 \end{tabular}
 & \begin{tabular}{c} 
 {\includegraphics[width=5cm]{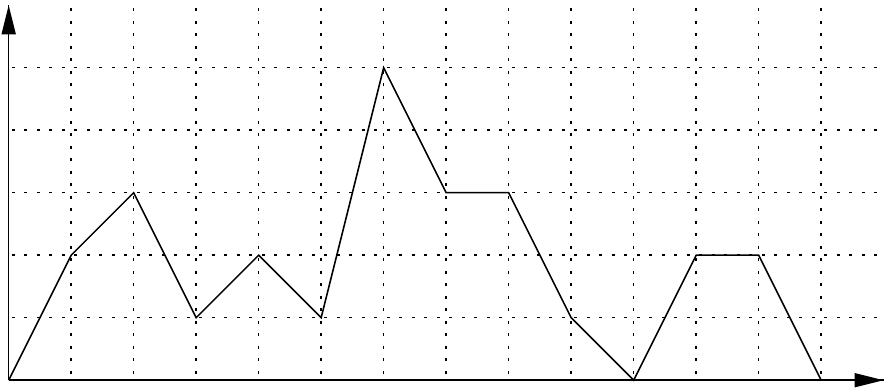}} 
\\ 
 excursion ($\cal E$)\\ 
 \end{tabular}\\
 \hline
 \end{tabular}
 \end{center}
 \caption{\label{fig-4types} 
 The four types of paths: walks, bridges, meanders, and excursions.
We refer to these walks as the Banderier--Flajolet model,  in contrast with the model in which we will consider a rational slope boundary.}
 \end{table*}
 
We restrict our attention to \emph{directed paths} which are defined by the fact that, for each jump $(x,y) \in \stepset$, one must have $x \geq  0$.  
The next definition allows to merge the probabilistic point of view (random walks) and the combinatorial point of view (lattice paths):
 \begin{Definition}
 For a given step set $\stepset = \{s_1,\ldots,s_m\}$, we define the respective {\it system of weights}
as $ \{w_1,\ldots,w_m\}$ where $w_j >0$ is the weight associated\- to step $s_j$ for $j=1,\ldots,m$. 
The {\it weight of a path} is defined as the product of the weights of its individual steps. 
 \end{Definition}

\smallskip
 This article mainly builds on the work done in~\cite{BaFl02}. Therein, the class of directed lattice paths 
in $\Z^2$ was investigated thoroughly by means of analytic combinatorics (see \cite{flaj09}). 
First, in Section 2, we give a bijection between lattice paths below a line of rational slope,
and lattice paths from the Banderier--Flajolet model.
In Section 3, we give Knuth's open problem on lattice paths with slope 2/5.
In Section 4, the needed bivariate generating function is defined and the governing functional equation is derived and solved:
here the ``kernel method'' plays the most significant r\^ole in order to obtain the generating function 
(as typical for many combinatorial objects which are recursively defined with a ``catalytic parameter'').
In Section 5, we tackle some questions on asymptotics, thus answering the question of Knuth.
In Section 6, we analyze what happens for the Duchon's club (lattice paths with slope 2/3), and other slopes.
In Section 7, we conclude with some open questions of computer algebra.

\section{A bijection for lattice paths below a rational slope}
\label{sec:bij}

Consider paths in the $\N^2$ lattice\footnote{We live in a world where $0\in\N$.} starting in the origin whose allowed steps are of the type either East or North (i.e.,~steps $(1,0)$ and $(0,1)$, respectively). Let $\alpha, \beta$ be positive rational numbers.
We restrict the walks to stay strictly below the barrier $L: y = \alpha x + \beta$. Hence, the allowed domain for our walks forms kind of an obtuse cone with the $x$-axis, the $y$-axis and the barrier $L$ as boundaries. The problem of counting walks in such a domain is equivalent to counting directed walks in the Banderier--Flajolet model, as seen via the following bijection:

\begin{proposition}{\rm [Bijection: Lattice paths below a rational slope are directed lattice paths.]}
	\label{prop:bijgen}
	Let $\Dc: y < \alpha x + \beta$ be the domain strictly below the barrier $L$. 
From now on, we assume without loss of generality that $\alpha=a/c$ and $\beta=b/c$ where $a, b, c$ are positive integers such that $\gcd(a,b,c)=1$  (thus, it may be the case that $a/c$ or $b/c$ are reducible fractions).
There exists a bijection between ``walks starting from the origin with North and East steps''
and ``directed walks starting from $(0,b)$ with the step set $\{(1,a), (1,-c)\}$''. What is more, the restriction of staying below the barrier $L$ is mapped to the restriction of staying above the $x$-axis. 	
\end{proposition}	
\begin{proof}
The following affine transformation gives the  bijection (see Figure \ref{fig:bij}):\vspace{-3mm}
	\begin{align*}
		 \begin{pmatrix} x \\y  \end{pmatrix} \mapsto
			\begin{pmatrix}
				x + y \\
				a x - c y + b
			\end{pmatrix}.
	\end{align*}
Indeed, the determinant of the involved linear mapping is $-(c+a) \neq 0$. 
What is more, the constraint of being below the barrier (i.e., one has $y<\alpha x+\beta$)
is thus forcing the new abscissa to be positive: $ax-cy+b>0$. The  gcd conditions ensures an optimal choice (i.e., the thinnest lattice) for the lattice on which walks will live. 
Note that this affine transformation gives a bijection not only in the case of initial step set North and East, but for any set of jumps.
\end{proof}

The purpose of this  bijection is to map walks of length $n$ 
to meanders (i.e.,~walks that stay above the $x$-axis) constructed out of $n$ unit steps into the positive $x$ direction. 

\begin{figure*}[t]
	\centering
	\subfloat[Rational slope model]{
		\includegraphics[width=0.5\textwidth]{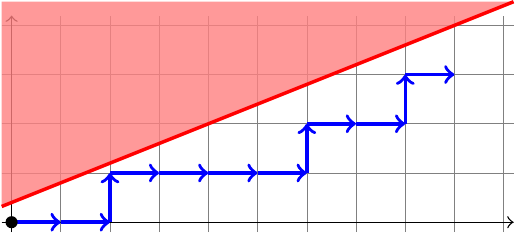}
	}
	\quad
	\subfloat[Banderier--Flajolet model]{
		\includegraphics[width=0.40\textwidth,height=41mm]{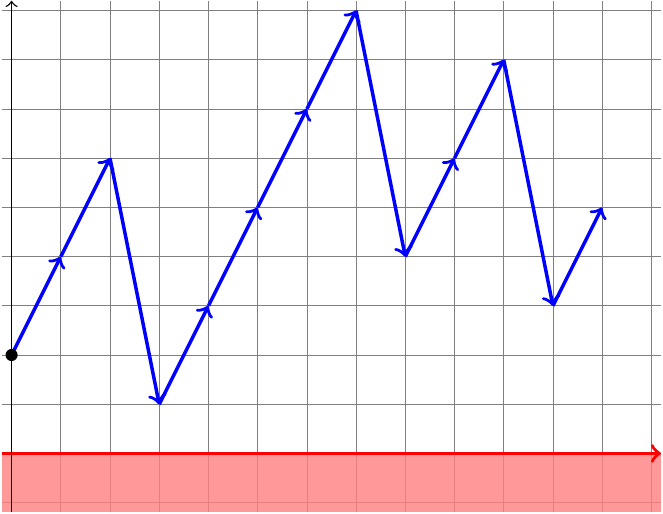}
	}
	\caption{Example showing the bijection from Proposition \ref{prop:bijgen}: Dyck paths below the line $y=2/5 x +2/5$ (or touching it)
	are in bijection with walks with jumps $+2$ and $-5$, starting at altitude 2, and staying above the line $y=0$ (or touching it).}
	\label{fig:bij}
\end{figure*}

Note that if one does not want the walk to touch the line $y=\alpha x + b/c$,
it corresponds to a model in which one allows to touch, but with a border at $y= \alpha x + (b-1)/c$.
Time reversal is also giving a bijection between walks starting at altitude $b$ with jumps $+a, -c$ and ending at 0
with walks starting at 0 and ending at altitude $b$ with jumps $-a,+c$.

\section{Knuth's AofA problem \#4}
During the conference ``Analysis of Algorithms'' (AofA'2014) in Paris in June 2014,
Knuth gave the first invited talk, dedicated to the memory of Philippe Flajolet (1948-2011).
The title of his lecture was ``Problems that Philippe would have loved'' and was pinpointing/developing five nice open problems with a good flavor of ``analytic combinatorics''
(his slides are available online\footnote{\tt http://www-cs-faculty.stanford.edu/$\sim$uno/flaj2014.pdf}).
The fourth problem was on ``Lattice paths of slope 2/5'', in which Knuth investigated Dyck paths under a line of slope 2/5, following the work of~\cite{Nakamigawa12}. 
This is best summarized by the two following slides of Knuth:

\begin{figure}[th]
	\begin{center}
		\fbox{\includegraphics[width=0.45\textwidth]{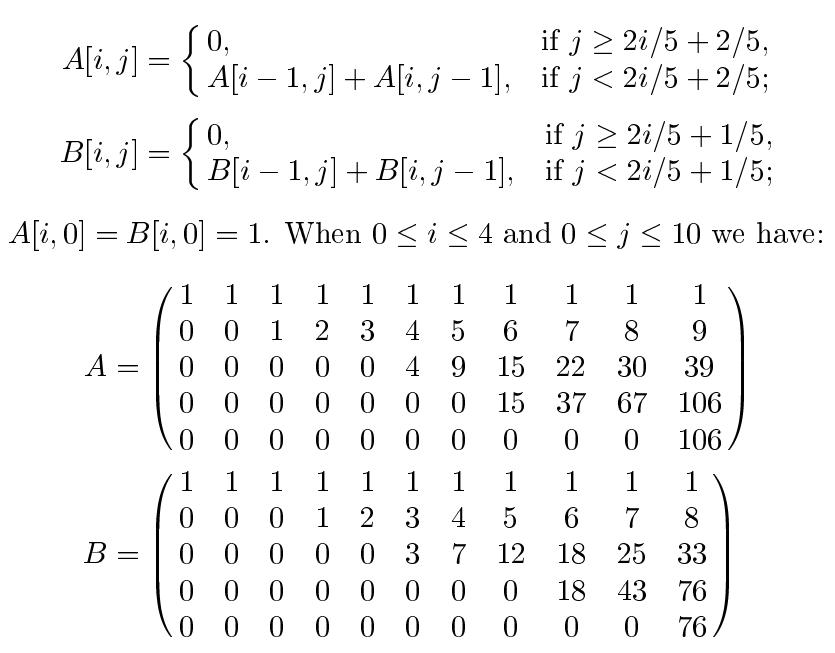}}
		\end{center}
\end{figure}
\begin{figure}[th]
	\begin{center}
		\fbox{\includegraphics[width=0.45\textwidth]{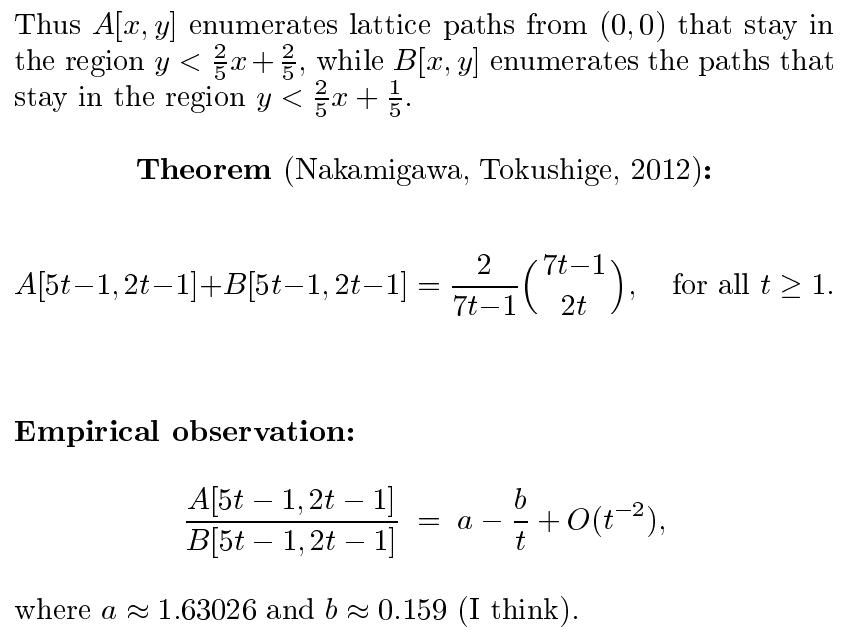}}
		\end{center}
\end{figure}

In the next sections we prove that Knuth was indeed right!  In order not to conflict with our notations, 
let us rename Knuth's constants $a$ and $b$ into $\kappa_1$ and $\kappa_2$.

\section{\large{Functional equation and closed forms for lattice paths of slope 2/5}}
\label{sec:func}

In this section, we show how to derive closed forms (i.e., explicit expressions) for the generating functions of lattice paths of slope 2/5 (and their coefficients).
First, define the jump polynomial $P(u) := u^{-2} + u^5$. Note that the bijection in Proposition \ref{prop:bijgen} gives jump sizes $2$ and $-5$. However, a time reversal gives this equivalent model, which is less technical to deal with (see below). 
Let $f_{n,k}$ be the number of walks of length $n$ which end at altitude $k$. The corresponding bivariate generating function is given by
\begin{align*}
	F(z,u) &= \sum_{n,k \geq 0} f_{n,k} z^n u^k = \sum_{n \geq 0} f_n(u) z^n = \sum_{k \geq 0} F_k(z) u^k,
\end{align*}
where the $f_n(u)$ encode all walks of length $n$ and the $F_k(z)$ are the generating functions for walks ending at altitude $k$. A step-by-step approach yields the following linear recurrence
\begin{align*}
	f_{n+1}(u) &= \{u^{\geq 0}\} \left[ P(u) f_n(u) \right] \qquad \text{ for } n \geq 0,
\end{align*}
with initial value $f_0(u)$ (i.e.,~the polynomial representing the walks of length $0$), and where $\{u^{\geq 0}\}$ is a linear operator extracting all the monomials in $u$ 
of nonnegative exponent. Summing the $z^{n+1} f_{n+1}(u)$ leads to  the functional equation
\begin{align}
	\label{eq:funceq}
	(1 - z P(u)) F(z,u) = f_0(u) - z u^{-2} F_0(z) - z u^{-1} F_1(z).
\end{align}
We apply the \emph{kernel method} in order to transform this equation into a system of linear equations for $F_0$ and $F_1$. The factor $K(z,u):=1-zP(u)$ is called the \emph{kernel} and the kernel equation is given by $K(z,u)=0$. Solving this equation for $u$, we obtain $7$ distinct solutions. These split into two groups, namely, we get $2$ small roots $u_1(z)$ and $u_2(z)$ (the ones going to 0 for $z\sim0$) and $5$ large roots which we call $v_i(z)$ for $i=1,\ldots,5$ (the ones going to infinity for $z\sim0$). It is legitimate to insert the $2$ small branches into~\eqref{eq:funceq} to obtain 
\begin{align*}
	z F_0 + z u_1 F_1 &= u_1^2 f_0(u_1),\\
	z F_0 + z u_2 F_1 &= u_2^2 f_0(u_2).
\end{align*}
This linear system is easily solved via Kramer's formula, which yields
\begin{align*}
	F_0(z) &= - \frac{u_1 u_2 \left(u_1 f_0(u_1) - u_2 f_0(u_2) \right)}{z (u_1 - u_2)}\,, \\
	F_1(z) &= \frac{u_1^2 f_0(u_1) - u_2^2 f_0(u_2) }{z (u_1 - u_2)}\,.
\end{align*}
Now, let the functions $F(z,u)$ and $F_k(z)$ denote functions associated with $f_0(u) = u^3$ (i.e.,~there is one walk of length $0$ at altitude $3$) and let the functions $G(z,u)$ and  $G_k(z)$ denote functions associated with $f_0(u) = u^4$. 
One thus gets the following theorem:
\begin{theorem}{\rm [Closed form for the generating functions]}
Let us consider walks in $\N^2$ with jumps $-2$ and $+5$.
The number of such walks starting at altitude 3 and ending at altitude 0 is given by $F_0(z)$,
 the number of such walks starting at altitude 4 and ending at altitude~1 is given by $G_1(z)$, and we have the following closed forms 
in the terms of the small roots $u_1(z)$ and $u_2(z)$ of $1-zP(u)=0$ with $P(u)=u^{-2}+u^5$:
\begin{align}
	\label{eq:F0G1}
	F_0(z) &=  - \frac{u_1 u_2 \left(u_1^4 - u_2^4\right)}{z (u_1 - u_2)}, \\
	G_1(z) &= \frac{u_1^6 - u_2^6 }{z (u_1 - u_2)}\,.
\end{align} 
\end{theorem}

Thanks to the bijection given in Section~2 between walks in the rational slope model and directed lattice paths in the Banderier--Flajolet model
(and by additionally reversing the time\footnote{Reversing the time allows us to express all the generating functions in terms of just 2 roots. If one does not  reserve time, 
everything is working well but is involving the 5 large roots, and gives more complicated closed forms.}), it is now possible to relate the quantities $A$ and $B$ of Knuth with $F_0$ and $G_1$:
\begin{align}
	\label{eq:ABdef}
	A_n&:= A[5n-1,2n-1] = [z^{7n-2}] G_1(z), \\
	B_n&:= B[5n-1,2n-1] = [z^{7n-2}] F_0(z).
\end{align}
Indeed, from the bijection of Prop.~2.1, the walks strictly below $y =(a/c) x + b/c$  
(with $a=2$, $c=5$) 
and ending at $(x,y) = (5n-1,2n-1)$
are mapped (in the Banderier--Flajolet model, not allowing to touch $y=0$) to walks 
starting at $(0,b)$ and ending at $(x+y,ax-cy+b)=(7n-2,3+b)$.
Reversing the time and allowing to touch $y=0$ (thus $b$ becomes $b-1$),
gives that $A_n$ counts walks starting at 4, ending at 1 (yeah, this is counted by $G_1$!) and that $B_n$ counts walks starting at~3, ending at 0 (yeah, this is counted by $F_0$!).
While there is no nice formula for $A_n$ or $B_n$ (see however~\cite{BanderierDrmota} for a formula involving nested sums of binomials),
it is striking that there is a simple, nice, formula for $A_n+B_n$:
\begin{theorem}{\rm [Closed form for the sum of coefficients]}
\label{theo:closedform}
The sum of the number of Dyck paths (in our rational slope model) touching or staying below $y=(2/5)x+1/5$ and $y=(2/5)x$   and ending on these lines
simplifies to the following expression:
\begin{align}
	\label{eq:AplusBex}
	A_n + B_n &= \frac{2}{7n-1} \binom{7n-1}{2n}.
\end{align}
\end{theorem}
\begin{proof}
A first proof of this was given by~\cite{Nakamigawa12} using a variant of the cycle lemma.
We can give another proof, indeed, our Theorem 4.1 implies that
\begin{align}
	\label{eq:AplusB}
	A_n + B_n &= [z^{7n-1}] \left( u_1^5 + u_2^5 \right)\,.
\end{align}
This suggests to use holonomy theory to prove the theorem. Indeed, first, a resultant equation gives the algebraic equation for $U:=u_1^5$ (namely, $z^7+(U-1)^5 U^2=0$)
and then the  Abel--Tannery--Cockle--Harley--Comtet theorem
(see the comment after Proposition 4 in~\cite{BanderierDrmota}) transforms it into a differential equation for 
the series $u_1^5(z^2)$. It is also the differential equation  (up to distinct initial conditions) for $u_2^5(z^2)$  (as $u_2$ is defined by the same equation as $u_1$), and thus of $u_1^5(z^2)+u_2^5(z^2)$.
Therefore, it directly gives the differential equation for the series $C(z)=\sum_n (A_n+B_n) z^n$,
and it corresponds to the following recurrence for its coefficients:
${C_{n+1}=\frac{7}{10}\frac{(7n+5)(7n+4)(7n+3)(7n+2)(7n+1)(7n-1)}{(5n+4)(5n+3)(5n+2)(5n+1)(2n+1)(n+1)}C_n\,,}$
which is exactly  the hypergeometric recurrence for $\frac{2}{7n-1} \binom{7n-1}{2n}$ (with the same initial condition).
This computation takes 1 second on an average computer, 
while, if not done in this way (e.g.,  if instead of the resultant shortcut above, one uses several {\tt{gfun[diffeq*diffeq]}} or variants of it in Maple),
the computations for such a simple binomial formula  surprisingly take hours.
\end{proof}

Some additional investigations conducted by Manuel Kauers (private communication) show that this is the only linear combination of $A_n$ and $B_n$ 
which leads to an hypergeometric solution
(to prove this, you can compute a recurrence for a formal linear combination  $r A_n+ s B_n$, 
and then check which conditions it implies on $r$ and $s$ if one wishes the associated recurrence to be of order 1, i.e., hypergeometric).
It thus appears that $r A_n+ s B_n$ is generically of order 5, with the exception 
of a sporadic $4A_n-B_n$ which is of order 4,
and the miraculous $A_n+B_n$ which is of order 1 (hypergeometric).

However, there are many other hypergeometric expressions floating around: expressions of the type of the right hand side of~\eqref{eq:AplusB} have nice hypergeometric closed forms.
This can also be explained in a combinatorial way, indeed we observe that setting $k=-5$ in Formula~(10) from~\cite{BaFl02},
 leads to $5 W_{-5}(z) = \Theta(A(z)+B(z))$ (where $\Theta$ is the pointing operator). The ``Knuth pointed walks'' are thus in 1-to-5 correspondence with unconstrained walks (see our Table 1, top left) ending at altitude -5.

Now, we need to do some analytic investigations in order to prove what Knuth conjectured:
\begin{align}
	\label{eq:aovberbconj}
	\frac{A_n}{B_n} & =   \kappa_1 - \frac{\kappa_2}{n} + \LandauO(n^{-2})
\end{align}
with $\kappa_1\approx 1.63026$ and $\kappa_2\approx 0.159$.

\section{Asymptotics}

As usual, we need to locate the dominant singularities, and to understand the local behavior there. 
Note that the fact that there are several dominant singularities makes the game harder here,
and this case was only sketched in~\cite{BaFl02}.
Similarly to what happens in the Perron--Frobenius theory, or in~\cite{BanderierDrmota}, 
a periodic behavior of the generating function leads to some more complicated proofs,
because additional details have to be taken into account. With respect to walks,
it is e.g.\ crucial to understand how singularities spread amongst the roots of the kernel.
To this aim,
some quantities will play a key r\^ole: the structural constant $\tau$ is defined as the unique positive root of $P'(\tau)$,
where $P(u)=u^{-2}+u^5$ is encoding the jumps, and the structural radius $\rho$ is given as $\rho = 1/P(\tau)$. For our problem, one thus has the explicit values:
\begin{align*}
	\tau    &= \sqrt[7]{\frac{2}{5}}, & 
	P(\tau) &= \frac{7}{10} \sqrt[7]{2^5 5^2}, &
	\rho    &= \frac{\sqrt[7]{2^2 5^5}}{7}.
\end{align*}

\begin{figure*}[t]
	\centering
\includegraphics[width=0.35\textwidth,height=51mm]{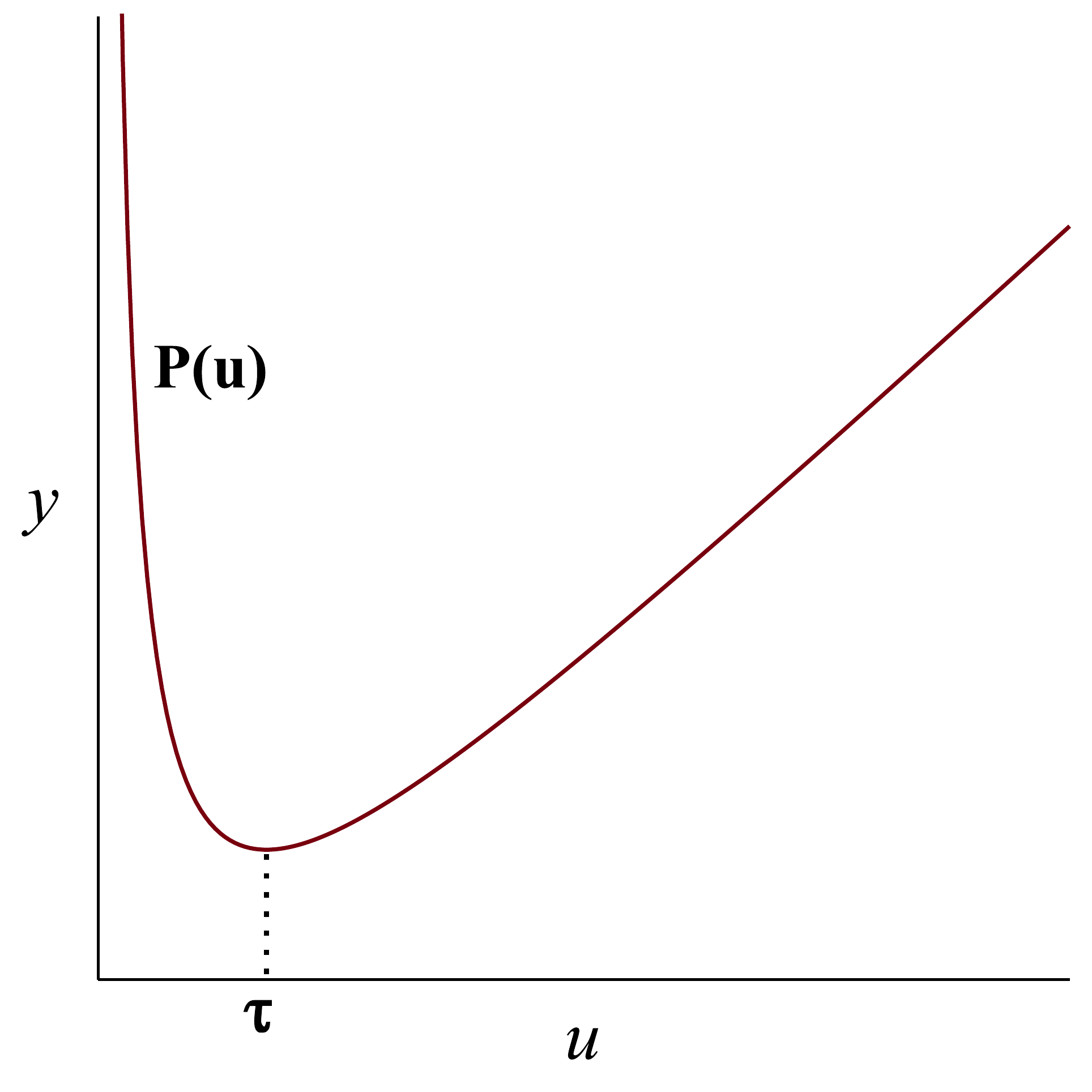} \qquad \includegraphics[width=0.35\textwidth,height=51mm]{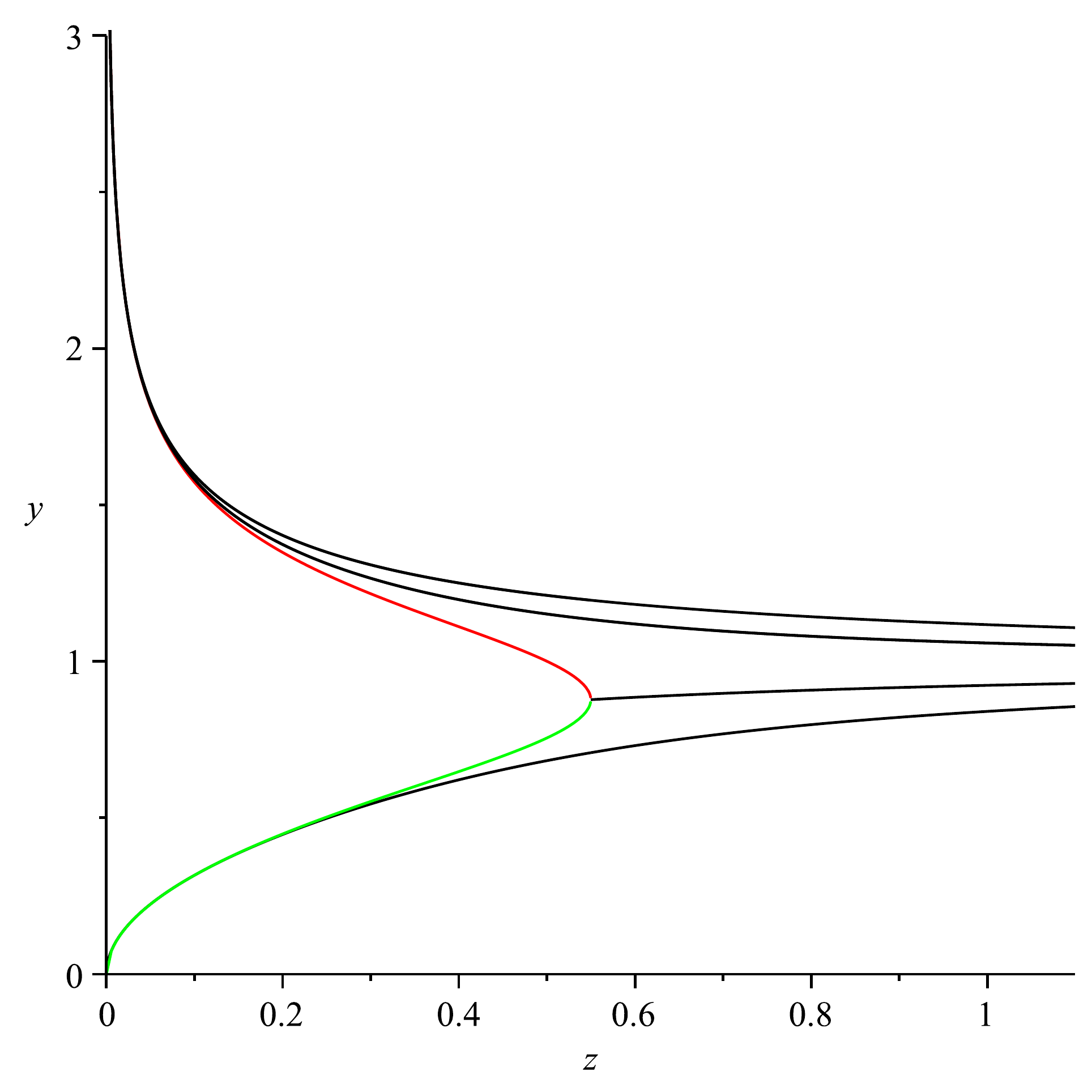}
\caption{$P(u)$ is the polynomial encoding the jumps, its saddle point $\tau$ gives the singularity $\rho=1/P(\tau)$ 
where the small root $u_1$ (in green) meets the large root $v_1$ (in red), with a square root behavior. (In black, we also plotted $|u_2|, 
|v_2|, |v_3|, |v_4|, |v_5|$.)
This is the key for all asymptotics of such lattice paths.
 }
\end{figure*}

From~\cite{BaFl02}, we know that the small branches $u_1$ and $u_2$ are possibly singular only at the roots of $P'(u)$. 
Note that the jump polynomial is periodic with period $p=7$ as $P(u) = u^{-2} H(u^7)$ with $H(u) = 1 + u$. Due to that, there are $7$ possible singularities of the small branches
\begin{align*}
	\zeta_k &= \rho \omega^k, \qquad \text{ with } \omega = e^{2 \pi i / 7}.
\end{align*}

Additionally, we have the following local behaviors:

\begin{lemma}
	\label{lem:u1u2}
	The limits of the small branches when $z \to \zeta_k$ exist and are equal to
	\begin{align*}
		u_1(z) \underset{z\,\sim \,\zeta_k}{=} &
				\begin{cases}
				\tau \omega^{-3k} + C_k \sqrt{1 - z / \zeta_k} +O((1 - z / \zeta_k)^{3/2}), \\
				\text{for } k = 2,5,7,\\     
					\tau_2 \omega^{-3k}   + D_k (1-z / \zeta_k) + O((1 - z / \zeta_k)^2), \\
					\text{for } k= 1,3,4,6,
				\end{cases}\\
		u_2(z) \underset{z\,\sim \,\zeta_k}{=}  &
				\begin{cases}
					\tau_2 \omega^{-3k}  + D_k (1-z / \zeta_k) + O((1 - z / \zeta_k)^2), \\
					\text{for } k = 2,5,7,\\     
					\tau \omega^{-3k} +  C_k \sqrt{1 - z / \zeta_k} +O((1 - z / \zeta_k)^{3/2}), \\
					\text{for } k= 1,3,4,6,
				\end{cases}
	\end{align*}
	where $\tau_2 = u_2(\rho)\approx -.707723271$ is the unique real root of
	$500t^{35}+3900t^{28}+13540t^{21}+27708t^{14}+37500t^7+3125$,
		where $C_k	   = -\frac{\tau}{\sqrt{5}} w^{-3k}$ and $D_k = \tau_2\frac{\tau_2^7+1}{5\tau_2^7-2} \omega^{-3k}$.
\end{lemma}
\begin{proof}
The Puiseux expansions of the small roots $u_i(z)$   have real coefficients, so one has $u_i(\bar z)=\overline{ u_i(z)}$,
what it is more, one has the following rotation law
	(for all $z \in \C$, with $|z| \leq \rho$ and $0<\arg(z)<\pi-2\pi/7$):
	\begin{align*}
	           u_1(\omega z) &=  \omega^{-3} u_2(z)\,\\ 
		u_2(\omega z) &= \omega^{-3} u_1(z). 
	\end{align*}  
Indeed, let us consider the function $U(z):=w^3 u_i(w z)$ and a mysterious quantity $X$,
defined by $X(z):= U^2-z\phi(U)$ (where $\phi(u):=u^2 P(u)$).
So we have $X(z)= (w^3 u_i(w z))^2-z \phi(w^3 u_i(w z))=  w^6 u_i(w z)^2 -z \phi(u_i(w z))$
(because $\phi$ is 7 periodic) 
and thus $wX(z/w)=  w(w^6 u_i(z)^2 -z/w \phi(u_i(z))) =  u_i(z)^2 -z \phi(u_i(z)))$,
which is 0 because we recognize here the kernel equation.
This implies that $X=U^2-z\phi(U)=0$ and thus $U$ is a root of the kernel.
Which one? It is one of the small roots, because it is converging to 0 at 0.
What is more, this root $U$ is not $u_i$, because it has a different Puiseux expansion 
(and  Puiseux expansion are  unique). So, by the analytic continuation principle (therefore, here, as far as we avoid the cut line 
$\arg(z)=-\pi$), we just proved that $w^3 u_1(w z)= u_2(z)$ and $w^3 u_2(w z)= u_1(z)$
(and this also proves a similar rotation law for large branches, but we do not need it).

Accordingly, at every $\zeta_k$, amongst the two small branches, only one branch becomes singular: this is $u_1$ for $k=2,5,7$ and $u_2$ for $k=1,3,4,6$. 
Hence, we directly see how the asymptotic expansion at the dominant singularities are correlated with the one of $u_1$ at $z=\rho=\zeta_7$,
which we derive following the approach of \cite{BaFl02}; this gives for $z\sim \rho$:
\begin{align*}
	u_{1}(z) &= \tau  + \coefb \sqrt{1 - z / \rho}  + \coefd (1-z/\rho)^{3/2}+\ldots,
\end{align*}	
where $\coefb = -\sqrt{2 \frac{P(\tau)}{P''(\tau)}}$. 
Note that in our case $P^{(3)}(\tau)=0$ (this funny cancellation holds for any $P(u)=p_{5} u^5+p_0+p_{-2} u^{-2}$ ), so even the formula for $\coefd$ is quite simple: $\coefd=-\frac{1}{2} \coefb$.  %A CREUSER! 

In the Lemma, the formula for $\tau_2=u_2(\rho)$ is obtained by a resultant computation.
\end{proof}

For the local analysis of the Knuth periodic generating functions $F_0(z)$ and $G_1(z)$, we introduce a shorthand notation:

\begin{Definition}{\rm [Local asymptotics extractor 	$\exzkn$]}
Let $F(z)$ be a generating function with~$p$ dominant singularities $\zeta_k$  (for $k=1,\ldots, p$). 
Define
\begin{align*}
	\exzkn F(z) &:= [z^n] \text{(Puiseux expansion of $F(z)$ at $z=\zeta_k$)}.
\end{align*}
\end{Definition}
\begin{proposition} 
Let $\rho$ be the positive real dominant singularity in the previous definition
(as $F$ is a generating function, it has real positive coefficients and therefore, by Pringsheim theorem, 
one of the $\zeta_k$'s  has to be real positive). 
When additionally the function $F(z)$ satisfies a rotation law $F(wz)=w^mF(z)$ (where $w=\exp(2i\pi/p)$), then one has a neat simplification (we relabel the $\zeta_k$'s such that  $\zeta_k:=w^k \rho$):
	\begin{align*}
 	 [z^n]F(z)-o(\rho^n) &=\sum_{k=1}^p \exzkn F(z) \\
 	                     &=\sum_{k=1}^p   \exzkn (w^m)^k   F(w^{-k} z) \\
 										 	 &=\sum_{k=1}^p (w^m)^k  (w^{-k})^n  [z^n]_\rho  F(z) \\
    									 &=\left(\sum_{k=1}^p   (w^k)^{m-n}\right)     [z^n]_\rho  F(z)\\
  										 &= p \, \chi_p(n-m)\, [z^n]_{\rho}  F(z), 
 \end{align*}
where $\chi_p(n)$ is 1 if $n$ is a multiple of $p$, 0 elsewhere.
\end{proposition}

We can apply this proposition to $F_0(z)$ and $G_1(z)$, because the rotation law for the $u_i$'s 
implies: $F_0(wz)=w^{-2} F_0(z)$ and $G_1(wz)=w^{-2} G_1(z)$. 
Thus, we just have to compute the asymptotics coming from the Puiseux expansion of  $F_0(z)$ and~$G_1(z)$ at $z=\rho$,
and multiply it by 7
(recall that it is classical to infer the asymptotics of the coefficients from the Puiseux expansion of the functions via the so-called ``transfer'' Theorem VI.3 from~\cite{flaj09}), this gives:

\pagebreak

\begin{theorem} The asymptotics for the number of excursions below $y=(2/5)x+2/5$ and $y=(2/5)x+1/5$  are given by:
\begin{align*}
	A_n &=[z^{7n-2}] G_1(z) = 
		 \alpha_1 \frac{\rho^{-7n}}{\sqrt{\pi (7n-2)^3}} + \ldots\\
	  & + \frac{3\alpha_2}{2}\frac{\rho^{-7n}}{\sqrt{\pi (7n-2)^5}} 
		+ \LandauO(n^{-7/2}), \\
	B_n &=[z^{7n-2}] F_0(z) = 
		 \beta_1 \frac{\rho^{-7n}}{\sqrt{\pi (7n-2)^3}} + \ldots\\
		& + \frac{3\beta_2}{2}\frac{\rho^{-7n}}{\sqrt{\pi (7n-2)^5}} 
		+ \LandauO(n^{-7/2}),
\end{align*}
with the following constant where we define the shorthand $\mu:=\tau_2/\tau$:
\begin{align*}
	\alpha_1 &= \frac{\mu^4 + 2 \mu^3 + 3 \mu^2 + 4 \mu + 5}{\sqrt{5}}, \\
	\beta_1 &= \sqrt{5}-\alpha_1, \\
	\alpha_2 &= -\frac{1}{10} \frac{5 \tau_2^7 (13\mu^4 + 22\mu^3 + 29\mu^2+36\mu+45)}{\sqrt{5}(5\tau_2^7-2)} + \ldots \\ 
	         & + \frac{2(15\mu^4+20\mu^3+13\mu^2-8\mu-45)}{\sqrt{5}(5\tau_2^7-2)},\\
	\beta_2 &=  -\frac{9}{10}\sqrt{5} - \alpha_2. \\
\end{align*}
\end{theorem} 
This theorem leads to the following asymptotics for $A_n+B_n$  
(and this is for sure a good sanity test, coherent with a direct application of Stirling's formula to the closed form formula~\eqref{eq:AplusBex} for $A_n+B_n$):
\begin{align*}
	A_n+B_n &= \sqrt{\frac{5}{7^3\pi}} \frac{\rho^{-7n}}{\sqrt{ n^3}} + \LandauO(n^{-5/2}).
\end{align*}
Finally, we directly get
\begin{align*}
	\frac{A_n}{B_n} &= \frac{\alpha_1 + \frac{3\alpha_2}{2 (7n-2)}}{\beta_1 + \frac{3\beta_2}{2 (7n-2)}} + \LandauO(n^{-2}) \\ 
					&= \frac{\alpha_1}{\beta_1} + \frac{3}{14} \left(\frac{\alpha_2 \beta_1 - \alpha_1 \beta_2}{\beta_1^2} \right)\frac{1}{n} + \LandauO(n^{-2}),
\end{align*}
which implies that Knuth's constants are 
\begin{align*}
	\kappa_1 &= \frac{\alpha_1}{\beta_1} = - \frac{5}{\mu^4 + 2 \mu^3 + 3 \mu^2 + 4 \mu} - 1 \\
	         &\approx 1.6302576629903501404248,\\          	\kappa_2 &= - \frac{3}{14} \left(\frac{\alpha_2 \beta_1 - \alpha_1 \beta_2}{\beta_1^2} \right) \\
          	         &= \frac{3}{9800} (13-236 \kappa_1 -194 \kappa_1^2 -388 \kappa_1^3 +437 \kappa_1^4) \\
          	         &\approx         0.1586682269720227755147.
\end{align*}
Few resultant computations give that $\kappa_1$ is the unique real root of the polynomial 
$23 x^5 -41 x^4 +10 x^3 -6x^2 - x- 1$, and $(7/3) \kappa_2$ 
is the unique real root of $11571875x^5-5363750x^4+628250x^3-97580x^2+5180x-142$.
The Galois group of each of these polynomial is $S_5$, this implies that there is no closed form formula for the Knuth constants $\kappa_1$ and $\kappa_2$ in terms of basic operations on integers, and root of any degree.

\bigskip
\section{Duchon's club and other slopes}

\begin{figure*}[t]
	\centering
	\subfloat[North-East model: Dyck paths below the line of slope 2/3]{
		\includegraphics[width=0.4\textwidth]{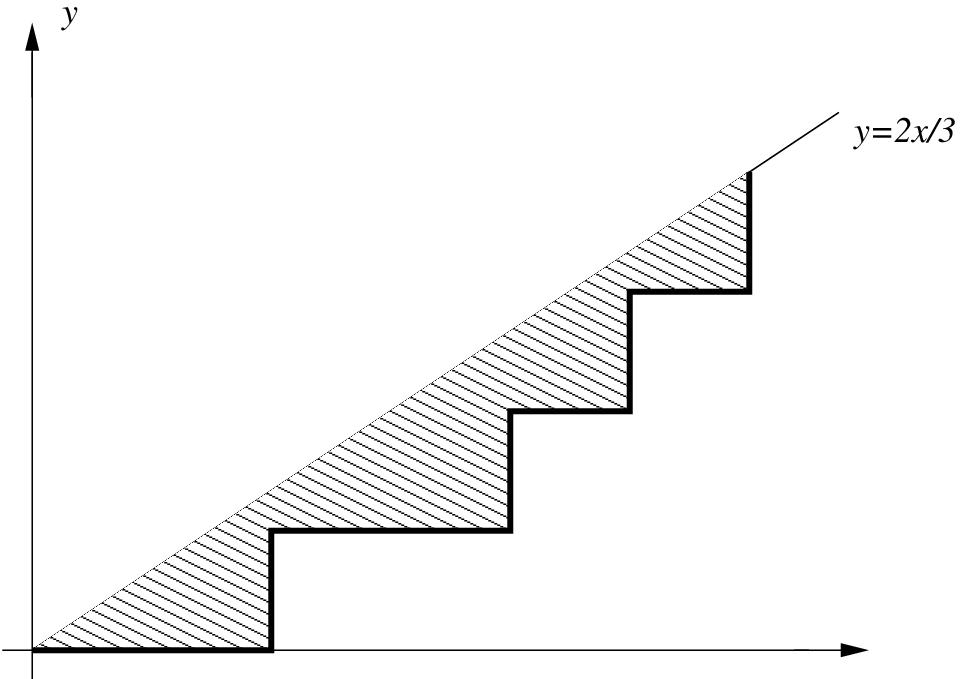}
	}
	\quad
	\subfloat[Banderier--Flajolet model: excursions with $+2$ and $-3$ jumps]{
		\includegraphics[width=0.5\textwidth]{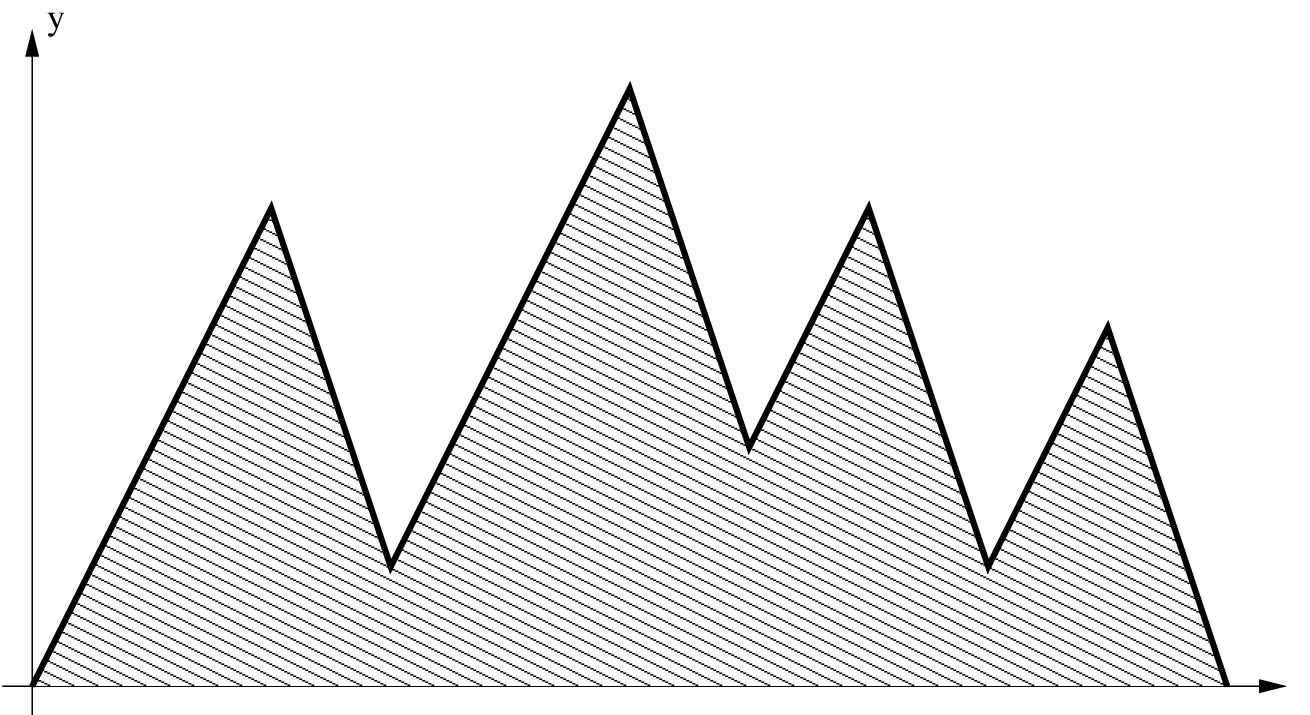}
	}
	\caption{Dyck paths below the line of slope 2/3 and Duchon's club histories (i.e., excursions with jumps $+2, -3$) are in bijection.
Duchon conjectured that the average area (in grey) was $K n^{3/2}$; 
our approach shows that $K=\sqrt{15 \pi}/2$.
 }
\end{figure*}

\smallskip
A Duchon walk is a Dyck path starting from $(0,0)$, with East and North steps,
and ending on the line $y=(2/3)x$. This model 
was analyzed in~\cite{Duchon},
and further investigated by~\cite{BaFl02}, 
who called it the ``Duchon's club'' model,
as it can be seen as the number of possible "histories" of couples entering the evening in a club\footnote{Caveat: there are no real life facts/anecdotes hidden behind this pun!}, and exiting by 3.
What is the number of possible histories (knowing the club is closing empty)?
Well, this is exactly the number $E_n$ of excursions with $n$ steps $+2, -3$,
 or (by reversal of the time) the number of excursions with $n$ steps $-2, +3$.
It was observed by Ernst Schulte-Geers,  and kindly communicated to us by Don Knuth,
that these numbers $E_n$ appeared already in an article by~\cite{Bizley54} 
(which gave some binomial formulas).
Duchon's club model should then be the Bizley--Duchon's club model;
Stigler's law of eponymy strikes again.

\smallskip
One open problem in the article~\cite{Duchon} was the following one:
``The mean area is asymptotic to $K n^{3/2}$, but the constant $K$ can
only be approximated to $3.43$''. 
Generalizing the approach of~\cite{BaGi06} (which was dealing with the "{\L}ukasiewicz walks", i.e., just one small root) 
to the more complex case of the kernel method involving several small roots,
we (together with Bernhard Gittenberger) show that $K=\sqrt{15 \pi}/2 \approx 3.432342124$,
using similar approaches that we used in the previous section (periodicity is again complicating the game).

\smallskip Last but not least, there is an equivalent of Theorem~\ref{theo:closedform} for any rational slope, and a nice expression in terms of the kernel roots
(we give more details in the full version of this article):

\pagebreak
\begin{theorem}\rm [``Additional'' closed forms]
	Let $a,c$ be integers such that $a<c$, and let $b$ be a multiple of $a$. 
	Let $A_s(k)$ be the number of Dyck walks below the line of slope $y=\frac{a}{c} x + \frac{k}{c}$, $k \geq 1$, ending at $(x_s,y_s)$ given by
	\begin{align*}
		x_s &= c s - 1, &
		y_s &= a s - 1.
	\end{align*}
	Then it holds for $s \geq 1$ and $\ell \in \N$ such that $(\ell + 1) a < c$ that
	\begin{align*}
		\sum_{k=\ell a+1}^{(\ell+1)a} A_s(k) &= \frac{\ell a+c}{(a+c)s+\ell-1} \dbinom{(a+c)s + \ell - 1}{as-1}.
	\end{align*}
\end{theorem}

\section{Conclusion}

We analyzed here some models of Dyck paths below a line of rational slope 2/5.
Besides the (pleasant) satisfaction of answering to a problem of Don Knuth,
this sheds light on properties of constrained lattice paths, including in the delicate case (for analysis) of a periodic  behavior.

The same approach extends (with some care) to any directed lattice paths (not only Dyck paths)
below a line of rational slope. This leads to some nice universal results,
we will present them in the full version of this article.

\smallskip
{\em En passant}, we encountered several computer algebra problems, we now list some of the corresponding questions, which should be written in details/implemented one day:
\begin{itemize}
\item The Flajolet--Salvy ACA algorithm (for analytic continuation of algebraic branches, see~\cite{flaj09}) is not yet available 
by default in most computer algebra systems, this makes them completely buggy while handling some ``RootOf'' algebraic functions: they don't follow the right branch. Human guidance remains therefore crucial for all the computer-algebra-assisted computations.
\item How to predict the minimal degree/order of the involved equations?  How to go efficiently from the differential equation to the algebraic equation, and conversely?
In the general case: the larger is the max of the numerator/denominator of the slope, 
 the more the computations with computer algebra are apocalyptic. However, there could be specific algorithms for these walk problems, as the corresponding functions have a lot of structure (often dictated by the kernel method).
 \item Is the platypus algorithm (see~\cite{BaFl02}) the fastest way to get the recurrence?
\item Is there a way to play (e.g.\ with Pad\'e approximations) in order to get universal behavior of the asymptotics, even in the case of an irrational slope? 
The nature of the constants then appearing in the asymptotics (and how to handle them efficiently) is open.
\end{itemize}

{\bf Acknowledgements:}
\label{sec:ack}
This work is the result of a collaboration founded by the SFB project F50 ``Algorithmic and Enumerative Combinatorics'' 
and the Franco-Austrian PHC ``Amadeus''. Michael Wallner is supported by the Austrian Science Fund (FWF) grant SFB F50-03 and by \"OAD, grant F04/2012. 
We also thank the two referees for their feedback, and not last but least,
Don Knuth and Manuel Kauers for exchanging on this problem.

\,\quad 

\addcontentsline{toc}{chapter}{References}
\bibliographystyle{plain}
\bibliography{knuthslope}

\begin{thebibliography}{10}

\bibitem{hexa}
Cyril Banderier, Mireille Bousquet-M{\'e}lou, Alain Denise, Philippe Flajolet,
  Dani{\`e}le Gardy, and Dominique Gouyou-Beauchamps.
\newblock Generating functions for generating trees.
\newblock {\em Discrete Math.}, 246(1-3):29--55, 2002.
\newblock Formal power series and algebraic combinatorics (Barcelona, 1999).

\bibitem{BanderierDrmota}
Cyril Banderier and Michael Drmota.
\newblock Formulae and asymptotics for coefficients of algebraic functions.
\newblock {\em Combin. Probab. Comput.}, 24(1):1--53, 2015.

\bibitem{BaFl02}
Cyril Banderier and Philippe Flajolet.
\newblock Basic analytic combinatorics of directed lattice paths.
\newblock {\em Theoret. Comput. Sci.}, 281(1-2):37--80, 2002.

\bibitem{bfss01}
Cyril Banderier, Philippe Flajolet, Gilles Schaeffer, and Mich{\`e}le Soria.
\newblock Random maps, coalescing saddles, singularity analysis, and {A}iry
  phenomena.
\newblock {\em Random Structures \& Algorithms}, 19(3-4):194--246, 2001.

\bibitem{BaGi06}
Cyril Banderier and Bernhard Gittenberger.
\newblock Analytic combinatorics of lattice paths: enumeration and asymptotics
  for the area.
\newblock {\em Discrete Math. Theor. Comput. Sci. Proc.}, AG:345--355, 2006.

\bibitem{Bizley54}
Michael Terence~Lewis Bizley.
\newblock Derivation of a new formula for the number of minimal lattice paths
  from {$(0,0)$} to {$(km,kn)$} having just {$t$} contacts with the line
  {$my=nx$} and having no points above this line; and a proof of {G}rossman's
  formula for the number of paths which may touch but do not rise above this
  line.
\newblock {\em J. Inst. Actuar.}, 80:55--62, 1954.

\bibitem{Fusy}
Mireille Bousquet-M{\'e}lou, {\'E}ric Fusy, and Louis-Fran{\c{c}}ois
  Pr{\'e}ville-Ratelle.
\newblock The number of intervals in the {$m$}-{T}amari lattices.
\newblock {\em Electron. J. Combin.}, 18(2):Paper 31, 26, 2011.

\bibitem{Duchon}
Philippe Duchon.
\newblock On the enumeration and generation of generalized {D}yck words.
\newblock {\em Discrete Math.}, 225(1-3):121--135, 2000.
\newblock Formal power series and algebraic combinatorics (Toronto, 1998).

\bibitem{Fayolle}
Guy Fayolle, Roudolf Iasnogorodski, and Vadim Malyshev.
\newblock {\em Random walks in the quarter-plane}, volume~40 of {\em
  Applications of Mathematics}.
\newblock Springer-Verlag, 1999.

\bibitem{flaj09}
Philippe Flajolet and Robert Sedgewick.
\newblock {\em Analytic Combinatorics}.
\newblock Cambridge University Press, 2009.

\bibitem{Kn69}
Donald~Ervin Knuth.
\newblock {\em The art of computer programming. {V}ol. 1: {F}undamental
  algorithms}.
\newblock Second printing. Addison-Wesley, 1969.

\bibitem{Mansour}
Toufik Mansour and Mark Shattuck.
\newblock Pattern avoiding partitions, sequence {A}054391 and the kernel
  method.
\newblock {\em Appl. Appl. Math.}, 6(12):397--411, 2011.

\bibitem{Nakamigawa12}
Tomoki Nakamigawa and Norihide Tokushige.
\newblock Counting lattice paths via a new cycle lemma.
\newblock {\em SIAM J. Discrete Math.}, 26(2):745--754, 2012.

\end{thebibliography}
\label{sec:biblio}
\end{document}